\documentclass{llncs}
\usepackage{times}

 \newtheorem{notation}{Notation}

\DeclareMathAlphabet{\mathsl}{OT1}{cmr}{m}{sl}

\newcommand\eg{\hbox{\textit{e.g.}}}
\newcommand\ie{\hbox{\textit{i.e.}}}

\newcommand{\false}{\texttt{f}}

\newcommand{\true}{\texttt{t}}

\newcommand{\defsymboldelta}{\stackrel{\Delta}  {=}}

\newcommand\Fifex[6]{#1;#2\mathrel{\buildrel #3\over{\hbox to #6pt{\rightarrowfill}}}#4;#5}

\newcommand\ra\rightarrow

\newcommand\tensor\otimes

\newcommand{\gdash}{\vdash\kern -4pt\vdash}

\newcount\PLv\newcount\PLw\newcount\PLx\newcount\PLy\newdimen\PLyy\newdimen\PLX
\newbox\PLdot \setbox\PLdot\hbox{\tiny.} \def\scl{.08} 
\def\PLot#1{\PLx`#1\advance\PLx-42\PLy\PLx\PLv\PLx\divide\PLy9\PLw\PLy\multiply
\PLw9\advance\PLx-\PLw\advance\PLx-4\PLy-\PLy\advance\PLy4\PLX=\the\PLx pt
\advance\PLyy\the\PLy pt\wd\PLdot=\scl\PLX\raise\scl\PLyy\copy\PLdot}
\def\draw#1{\ifx#1\end\let\next=\relax\else\PLot#1\let\next=\draw\fi\next}

\def\invamp{\hbox{\PLyy=70pt\draw :::;DMV_gqppyyyyyooooxxxnnwvlutkjaWNE=5-./9%
9:::CCCC:::99/..--544=EENWWaajjjkktttttttNNNVVVVVVVV\end \hskip4pt}}
\newbox\iabox\setbox\iabox\invamp

\long\def\hide#1\endhide{}

\DeclareMathAlphabet{\mathsl}{OT1}{cmr}{m}{sl}

\newcommand{\rTell}{\rm R_{TELL}}
\newcommand{\rTellCS}{\rm R'_{TELL}}

\newcommand{\rSum}{\rm R_{SUM}}
\newcommand{\rSumCS}{\rm R'_{SUM}}

\newcommand{\rBang}{\rm R_{!}}

\newcommand{\rLocal}{\rm R_{LOC}}
\newcommand{\rLocalCS}{\rm R'_{LOC}}

\newcommand{\rUnless}{\rm R_{Un}}
\newcommand{\rObserv}{\rm R_{Obs}}
\newcommand{\rCall}{\rm R_{CALL}}
\newcommand{\rCallCS}{\rm R'_{CALL}}

\newcommand{\rEquiv}{\rm R_{EQUIV}}

\newcommand{\vx}{\overline{x}}

\newcommand{\vy}{\overline{y}}

\newcommand{\Barb}[2]{#1\Downarrow_{#2}}
\newcommand{\BarbC}[2]{#1\downdownarrows_{#2}}

\def\cC{\mathcal{C}}
\def\cD{\mathcal{D}}

\newcommand{\fv}{{\it fv}}

\newcommand{\bv}{{\it bv}} 
\newcommand{\tellp}[1]{\mathbf{tell}(#1)}

\newcommand{\whenp}[2]{\mathbf{ask} \  (#1) \  \mathbf{then} \ #2}

\newcommand{\nextp}[1]{ \mathbf{next}\ #1}
\newcommand{\unlessp}[2]{\mathbf{unless} \  (#1) \  \nextp{#2}}
\newcommand{\localp}[2]{(\mathbf{local} \, #1)  \, #2}

\newcommand{\bangp}[1]{! #1}

\newcommand{\skipp}{\mathbf{skip}}

\newcommand{\defsymbol}{\stackrel{\Delta}  {=}}

\newcommand{\TCC}{\texttt{tcc}}

\newcommand{\ccp}{CCP}

\newcommand{\equivP}{\cong}

\newcommand{\entails}{\models}

\long\def\comment#1{}

\newcommand{\redi}{\longrightarrow}
\newcommand{\redirex}{\longrightarrow^{*}}

\newcommand{\rede}[1]{\stackrel{\,\,#1\,\,}{\,\,
=\hspace{-0.1cm}=\hspace{-0.1cm}=\hspace{-0.1cm}\Longrightarrow}} 

\newcommand{\rediIdx}[1]{\stackrel{\,\,\,[#1]_{}\,\,}{\,\, \longrightarrow}} 

\newcommand{\rediIdxJ}[2]{\xrightarrow{[#1]_{#2}}}

\newcommand{\lub}{\it lub}

\newcommand{\equivC}{\cong}

\newcommand{\sliced}{\bullet}

\newcommand{\startc}{\texttt{start}}
\newcommand{\stopc}{\texttt{stop}}
\newcommand{\beatc}{\texttt{beat}}

\newcommand{\past}[1]{\circleddash#1}

\newcommand{\fromCollect}[1]{\lceil #1 \rceil}
\newcommand{\toCollect}[1]{\lfloor #1 \rfloor}
\newcommand{\appr}{\preceq^{\sharp}}

\newcommand{\idxP}[1]{\!:\!#1}

\usepackage{subfig}
\usepackage{fancyvrb}
\usepackage{graphicx}
\usepackage{wrapfig}

\usepackage{color}
\usepackage{latexsym}
\usepackage{amssymb}
\usepackage{amsmath}
\usepackage{proof} 
\usepackage{listings}
\usepackage{lineno}

\usepackage{url}
\usepackage{verbatim}
\usepackage[linesnumbered]{algorithm2e}

\begin{document}
\long\def\comment#1{}


\title{Slicing Concurrent Constraint Programs}

\author{ M. Falaschi\inst{1} ~~~ M. Gabbrielli\inst{2} ~~~ C. Olarte\inst{3} ~~~ C. Palamidessi \inst{4}}
\institute{Dipartimento di Ingegneria dell'Informazione e Scienze Matematiche\\
              Universit\`a di Siena, Italy.\\
		      \email{moreno.falaschi@unisi.it.} \and
	 Dipartimento di Informatica - Scienza e Ingegneria, 
	 Universit\`a di Bologna, Italy.
	 \email{gabbri@cs.unibo.it.}
	 \and
	 ECT, Universidade Federal do Rio Grande do Norte, Brazil\\
	 \email{carlos.olarte@gmail.com.}
	 \and
	 INRIA and LIX, \'{E}cole Polytechnique, France.\\
	 \email{catuscia@lix.polytechnique.fr.}
	   }


	   \date{}
\maketitle


\makeatletter{}
\begin{abstract}

Concurrent Constraint Programming (\ccp) is a declarative model for
concurrency where agents interact by telling and asking constraints
(pieces of information) in a shared store.  Some previous works have
developed (approximated) 
declarative debuggers for CCP languages. However, the task of
debugging 
concurrent programs remains difficult. 
In this paper we define a dynamic slicer for \ccp\ and we show it to
be a useful companion tool for the existing debugging techniques. 
We start with a partial computation (a trace) that 
shows the presence of bugs. Often, the quantity of information in such
a trace is
overwhelming, and the user gets easily lost, since she cannot
focus on the sources of the bugs. Our slicer allows for marking part
of the 
state of the computation and assists the user to eliminate most of
the redundant 
information in order to highlight the errors. We show that this
technique can be 
tailored to timed variants of CCP.  We also develop a prototypical 
implementation freely available for making experiments. 
\end{abstract}


\makeatletter{}
\section{Introduction}\label{sectionintroduction}

Concurrent constraint programming (CCP)  \cite{cp-book,saraswat91popl}  
(see a survey in 
\cite{DBLP:journals/constraints/OlarteRV13}) 
  combines   concurrency primitives with the ability 
to deal with constraints, and hence, with partial information. 
The notion of concurrency 
is based upon the shared-variables communication model.
CCP is intended for reasoning, modeling and 
programming 
concurrent  agents (or processes) that interact with each other and
their environment by posting and asking information in a medium, 
a so-called store.
Agents in  CCP can be seen
as both computing processes (behavioral style) and as logic formulae
(declarative style). 
Hence CCP can exploit reasoning techniques from both 
process calculi and logic. 

CCP is a very flexible model and then,  it
has been applied to an increasing number of different 
fields such as 
probabilistic and stochastic \cite{bortolussi}, timed 
\cite{DBLP:journals/jsc/SaraswatJG96,NPV02,DBLP:journals/iandc/BoerGM00} and mobile  \cite{Olarte:08:SAC} systems. More recently, \ccp\ languages have been proposed for the specification of spatial and epistemic behaviors as in, e.g.,  
social networks 
 \cite{DBLP:conf/concur/KnightPPV12,DBLP:journals/tcs/OlartePN15}. 

One crucial problem when working with a concurrent language is being 
able to provide tools to debug  programs. This is particularly 
useful for a language in which a program can 
generate a 
large number of parallel running agents. 
In order to tame this complexity, abstract interpretation
techniques 
 have been 
considered  (e.g. in 
\cite{CFM94,CominiTV11absdiag,DBLP:journals/tplp/FalaschiOP15}) as well as  
(abstract) declarative
debuggers   following the seminal work of 
Shapiro \cite{Shapiro83}.
However, these techniques 
are approximated (case of abstract interpretation) or it can be difficult to apply   them when dealing with complex programs (case of 
declarative debugging). It would be 
useful to have a semi automatic tool able to interact with 
the user and filter, in a given computation,  the information which is 
relevant to a particular observation or result.
In other words, the programmer could mark the outcome that she is 
interested to check in a particular  computation that she 
suspects to be wrong. Then, a corresponding 
depurated partial computation is obtained automatically, 
where only the information relevant to
the marked parts is present.

Slicing was introduced in some pioneer works by Mark Weiser \cite{MW84}. 
It was originally defined as a static technique, independent of  any particular input of the program. 
Then, the technique was extended by introducing the so called 
dynamic program slicing \cite{KL88}. This technique
is useful for simplifying the debugging process, by selecting a portion of
the program containing 
the faulty code.
Dynamic program slicing has been applied
to several programming paradigms, for instance to imperative 
programming \cite{KL88}, functional programming 
\cite{OSV08}, 
Term Rewriting \cite{ABER11},
and functional logic programming \cite{ABFR14}. 
The reader may refer to \cite{Silva2012} for a  survey. 
%
%

In this paper we present
the first formal framework for CCP 
dynamic slicing and show, by some working examples and a  prototypical tool, 
the main features of this approach. Our aim is to  help  the
programmer to debug her program, in cases where she could not find
the bugs by using other debuggers.
We proceed with three main steps.
First we extend the standard operational semantics of \ccp\ 
to a  ``collecting semantics'' that adds the needed information for  the slicer. 
Second, we propose several analyses of the faulty 
situation based on error symptoms, 
including causality, variable dependencies, unexpected behaviors and 
store inconsistencies. Thirdly, 
we define a marking algorithm of the redundant items and define a trace slice. 
Our algorithm is flexible and it can 
deal with different variants of CCP. 
In particular, we show  how to apply it to 
timed extensions of \ccp\ \cite{DBLP:journals/jsc/SaraswatJG96}.

\emph{Organization. } Section \ref{sectionccp} 
describes CCP and its  operational semantics.
In Section \ref{sectionslicing} we introduce  a slicing technique for CCP. In Section   \ref{sectionapplications} we extend our method to consider timed  CCP programs. We  present a working 
prototypical implementation of the slicer available at \url{http://subsell.logic.at/slicer/}. We describe an example  using the slicer to debug a multimedia interacting system programmed in timed CCP. 
Due to lack of space, other examples are given only in the web page of the tool as, for instance, a biochemical system specified in timed CCP. Finally, Section \ref{sectionconclusions} concludes.

\makeatletter{}

\section{Concurrent Constraint Programming}\label{sectionccp}

 
Processes in \ccp\ \emph{interact} with each other by \emph{telling} and \emph{asking} 
constraints (pieces of information) in a common 
store of partial information. The type of constraints
 is not fixed but parametric in a constraint system (CS). 
 Intuitively, a
CS  provides a signature from which  constraints
can be built from basic tokens (e.g., predicate symbols), and two basic operations: conjunction ($\sqcup$) and variable hiding ($\exists$). The CS defines also an
\emph{entailment} relation ($\entails$) specifying inter-dependencies
between constraints:  $c\entails d$  means that the
information $d$ can be deduced from the information 
$c$. 
Such systems  can be formalized as a Scott information system as in \cite{saraswat91popl},
as  cylindric algebras \cite{BoerPP95}, or they can be built upon a suitable fragment of logic \eg, as in \cite{NPV02}. 
Here we 
follow \cite{BoerPP95}, since the other approaches 
can be seen as an instance of this definition. 

\begin{definition}[Constraint System --CS--] \label{def:cs}
A cylindric constraint system  is a structure 
 $
{\bf C} = \langle \cC,\leq,\sqcup,\true,\false,{\it Var}, \exists,
D \rangle
$ s.t. 
\\\noindent{-}
 $\langle \cC,\leq,\sqcup,\true,\false \rangle$ is a complete 
 algebraic lattice
with $\sqcup$ the $\lub$ operation (representing 
\emph{conjunction}).
Elements in $\cC$ are called \emph{constraints}
with typical elements $c,c',d,d'...$, 
and $\true$, $\false$ the least and the greatest
elements. 
If $c\leq d$, we say that $d$ entails $c$ and we write $d \entails c$. 
If $c\leq d$ and $d\leq c$ we write $c \cong d$. 
\\\noindent{-}${\it Var}$ is a denumerable set of variables and for each
$x\in {\it Var}$ the function $\exists x: \cC \to \cC$ is a
cylindrification operator satisfying:
		    (1) $\exists x (c) \leq c$. 
		(2) If $c\leq d$ then $\exists x (c) \leq \exists x (d)$.
		(3) $\exists x(c \sqcup \exists x (d)) \equivC \exists x(c) \sqcup \exists x(d)$.
		(4) $\exists x\exists y(c) \equivC \exists y\exists x (c)$.
		(5) For an increasing chain $c_1 \leq c_2 \leq c_3...$, $ \exists x
\bigsqcup_i c_i \equivC \bigsqcup_i \exists x ( c_i) $.
\\\noindent{-} For each $x,y \in {\it Var}$, the constraint $d_{xy} \in D$ is a
\emph{diagonal element} and it satisfies:
			(1) $d_{xx} \equivC \true$.
		(2) If $z$ is different from $x,y$ then $d_{xy} \equivC \exists z(d_{xz}
\sqcup d_{zy})$.
		(3) If $x$ is different from $y$ then $c \leq d_{xy} \sqcup
\exists x(c\sqcup d_{xy})$.
\end{definition}
The cylindrification operator models a sort of existential
quantification  for hiding information. As usual, $\exists x. c$ binds $x$ in $c$. We use $\fv(c)$ (resp. $\bv(c)$) to denote the set of free (resp. bound) variables in $c$. 
The diagonal element $d_{xy}$ can be thought of as the equality $x=y$, 
useful to define  substitutions of the form $[t/x]$ (see the details, e.g., 
in \cite{DBLP:journals/tplp/FalaschiOP15}). 

 As an example, consider the finite domain constraint system  (FD) \cite{HentenryckSD98}. This system assumes variables to range over finite domains and, in addition to equality, one may have predicates  that restrict the possible values of a variable as in  $x<42$. 
%
%

\subsection{The language of \ccp\ processes}
In the spirit of process calculi, the language of processes in   \ccp\  is given by a small number of primitive operators or combinators as described below. 
\begin{definition}[Syntax of Indeterminate \ccp\ \cite{saraswat91popl}]\label{def:syntax-lcc} 
 Processes in \ccp\  are built from constraints in the underlying constraint system and  the syntax:

$
P,Q ::= \skipp\mid \tellp{c} \mid   
\sum\limits_{i\in I}\whenp{c_i}{P_i} \mid   
P \parallel Q  \mid 
\localp{x}{P}   \mid  
p(\overline{x})
$
\end{definition}

The process $\skipp$ represents inaction. The process 
 $\tellp{c}$ adds $c$ to the current store $d$ producing the new store $c\sqcup d$.
 Given a non-empty finite set of indexes $I$, the process $\sum\limits_{i\in I}\whenp{c_i}{P_i}$ non-deterministically chooses  $P_k$ for execution if the   store entails $c_k$. The chosen alternative, if any, precludes the others.
This provides a powerful synchronization mechanism based on constraint entailment.   
 When $I$ is a singleton, we shall omit the ``$\sum$'' and we simply write $\whenp{c}{P}$. 
 
The process $P\parallel Q$ represents the parallel (interleaved) execution of $P$ and $Q$. 
The process 
$\localp {x}{P}$ behaves as $P$ and binds the variable 
$x$ to be local to it. We  use 
$\fv(P)$ (resp. $\bv(P)$) to denote  the set of free (resp. bound) variables in 
$P$. 

Given a process definition  $p(\overline{y}) \defsymboldelta P$,  
where all free variables of $P$ are in the set of pairwise distinct
variables $\overline{y}$, the process $p(\overline{x})$  evolves into 
$P[\overline{x}/\overline{y}]$.  A \ccp\ program takes the form 
$\cD.P$ where $\cD$ is  a set of  process definitions and $P$ is a process.

The Structural Operational Semantics (SOS)   of \ccp\ is given by the transition relation $ \gamma \redi \gamma'$  
satisfying the rules in Fig. \ref{fig:sos}.
Here we follow the formulation 
 in \cite{fages01ic} where the local variables created by the program appear explicitly in the transition system and parallel composition of agents is identified to a multiset of agents. 
More precisely, a \emph{configuration} $\gamma$ is a triple of the  form  
$(X;  \Gamma ;  c)$, where $c$ is a constraint representing the  store,  $\Gamma$ is a multiset of processes,
and $X$ is a set of hidden 
(local) variables of $c$ and $\Gamma$. The multiset $\Gamma=P_1,P_2,\ldots,P_n$  
represents the process  $P_1 \parallel P_2 \parallel \cdots \parallel P_n$. We shall indistinguishably
use both notations to denote parallel composition. Moreover, processes  are quotiented by a structural congruence relation $\equivP$  satisfying: 
 (STR1) $P \equivP Q$ if they differ only by a renaming of bound variables
   (alpha conversion);
 (STR2) $P\parallel Q \equivP Q \parallel P$;
 (STR3) $P \parallel (Q \parallel R) \equivP (P \parallel Q)  \parallel R$; (STR4) 
 $P \parallel \skipp \equivP P$.

Let us briefly explain the rules  in Figure \ref{fig:sos}. A tell agent $\tellp{c}$ adds $c$ to the current store $d$ (Rule $\rTell$); 
the process $\sum\limits_{i\in I}\whenp{c_i}{P_i}$ executes $P_k$ if its corresponding guard $c_k$ 
can be entailed from the store (Rule $\rSum$); a local process $\localp{x}{P}$ adds $x$ to the set of hidden variable $X$ when no clashes of variables occur (Rule $\rLocal$). Observe that Rule $\rEquiv$ can be used   to do alpha conversion if the premise 
of $\rLocal$
 cannot be satisfied; the call $p(\vx)$ executes the body of the process definition (Rule $\rCall$).

\begin{definition}[Observables]\label{def:obs}
Let $\redirex$ denote the  reflexive and transitive closure of $\redi$. If 
$(X;\Gamma; d) \redirex(X';\Gamma';d')$ and 
$\exists X'. d' \entails c$ we write 
$\Barb{(X;\Gamma;d)}{c}$.
If $X=\emptyset$ and $d=\true$   we simply write   $\Barb{\Gamma}{c}$.
\end{definition}

Intuitively, if $P$ is a process then 
$\Barb{P}{c}$ says that $P$ can reach a store $d$ strong enough to entail $c$, \ie, $c$ is an output of $P$. Note that the variables in $X'$ above are hidden from $d'$  since the information about  them is not observable.


\begin{figure}[t]
\resizebox{.92\textwidth}{!}
{
$
\begin{array}{ccc}
\infer[\rTell]
{(X; \tellp{c},\Gamma;d) \redi (X;\skipp, \Gamma;c\sqcup d)}
{}
\qquad
\infer[\rSum]{
  (X;\sum\limits_{i\in I}\whenp{c_i}{P_i},\Gamma;d) \redi (X;P_k,\Gamma;d)
  }
  {
  d \entails c_k \quad k\in I
  }
 \\\\
\infer[\rLocal]
{(X;\localp{x}{P},\Gamma;d) 
\redi (X\cup\{x\};P,\Gamma;d)
}
{x \notin X \cup fv(d) \cup fv(\Gamma)}
\qquad
\infer[\rCall]
{(X;p(\vx),\Gamma;d) \redi  (X;P[\vx/\vy],\Gamma;d) }
{p(\vy) \defsymboldelta P  \in \cD } \\\\
\infer[\rEquiv]
{(X;\Gamma;c) \redi (Y;\Delta;d)}
{(X;\Gamma;c) \equivP (X';\Gamma';c')  \redi (Y';\Delta';d') \equivP (Y;\Delta;d)}
\end{array}
$
}
\caption{Operational semantics for \ccp\ calculi\label{fig:sos}}
\end{figure}

\makeatletter{}
\section{Slicing a CCP program}\label{sectionslicing}

Dynamic slicing is a technique that helps the user to debug her program  by 
simplifying a partial execution trace, thus depurating it from parts which
are irrelevant to find the bug. 
It can also help to highlight parts of the programs which have 
been wrongly ignored by the execution of a wrong piece of code.

Our slicing technique  consists of three main steps:
\begin{enumerate}
 \item[{\bf S1}] \emph{Generating a (finite) trace} 
 of the program. For that, we propose a \emph{collecting semantics}  that generates the (meta) information needed for the slicer. 
 \item[{\bf S2}] \emph{Marking the final store}, to choose some of the constraints that, according to the symptoms detected, should or should not be in the final store. 
 \item[{\bf S3}] \emph{Computing the trace slice}, to select the processes and constraints  that were relevant to produce the (marked) final store. 
\end{enumerate}

\subsection{Collecting Semantics (Step ${\bf S1}$)}\label{subsec:collect}
The slicer we propose requires some extra information  from the execution of the processes. More precisely,    (1) in each operational step $\gamma \to \gamma'$, we need to highlight the process that was reduced; and (2) the constraints accumulated in the store must reflect, exactly, the contribution of each process to the store. 

In order to solve (1) and (2), 
we propose a collecting semantics that extracts the needed meta information 
 for the slicer. The  rules are in Figure \ref{fig:colsem} and explained 
below. 

\begin{figure}[t]
\begin{center}
\resizebox{.55\textwidth}{!}
{
$
\begin{array}{ccc}
\infer[\rTellCS]
{(X; \Gamma, \tellp{c}\idxP{i},\Gamma';S) \rediIdxJ{i}{} (X \cup Y;\Gamma,\Gamma'; S \cup S_c)}
{\langle Y, S_c\rangle = atoms(c, {\it fvars})}
\\\\
\infer[\rSumCS]{
  (X;\Gamma, \sum\limits_{l\in I}\whenp{c_l}{P_l}\idxP{i},\Gamma';S) \rediIdxJ{i}{k} (X;\Gamma,P_k\idxP{j},\Gamma';S)
  }
  {
  \bigsqcup\limits_{d\in S} d \entails c_k \quad k\in I
  }
 \\\\
\infer[\rLocalCS]
{(X;\Gamma, \localp{x}{P}\idxP{i},\Gamma';S) 
\rediIdxJ{i}{} (X\cup\{x'\};\Gamma,P[x'/x]\idxP{j},\Gamma';S)
}
{x' \in Var \setminus {\it fvars}}
\\\\
\infer[\rCallCS]
{(X;\Gamma, p(\vx)\idxP{i},\Gamma';S) \rediIdxJ{i}{}  (X;\Gamma, P[\vx/\vy]\idxP{j},\Gamma';S) }
{p(\vy) \defsymboldelta P  \in \cD } \\\\
\end{array}
$
}
\end{center}
\caption{Collecting semantics for \ccp\ calculi. $\Gamma$ and $\Gamma'$ are (possibly empty) sequences of processes.
${\it fvars} = X\cup fv(S) \cup fv(\Gamma) \cup fv(\Gamma')$. In ``:$j$'', $j$ is a fresh identifier. 
\label{fig:colsem}}
\end{figure}

The semantics considers configurations of the shape $(X ; \Gamma ; S)$
where $X$ is  a set of hidden variables, $\Gamma$ is a 
sequence of processes with \emph{identifiers} and $S$ is a set of atomic constraints. 
Let us explain the last two components. 
We identify  the  parallel composition 
 $Q = P_1 \parallel  \cdots \parallel P_n$ 
 with the   \emph{sequence}  $\Gamma_{Q} =P_1\idxP{i_1}, \cdots, P_n\idxP{i_n}$
 where $i_j   \in \mathbb{N}$ is a unique identifier for $P_j$.  Abusing of the notation, we usually write $Q\idxP{i}$ instead of $\Gamma_Q$ when the indexes in the parallel composition are unimportant. Moreover, we shall use $\epsilon$ to denote an empty sequence of processes. 
The  context $\Gamma, P\idxP{i},\Gamma'$ represents   that $P$ is preceded and followed, respectively, 
by the (possibly empty) sequences of processes $\Gamma$ and $\Gamma'$.  The use of indexes will allow us to distinguish, e.g.,   the three different  occurrences of $P$ in   
  ``$\Gamma_1,P\idxP{i},\Gamma_2,P\idxP{j}, (\whenp{c}{P})\idxP{k}$''. 

Transitions are labeled  with $\rediIdxJ{i}{k}$ where $i$ is the identifier of the reduced process  and $k$ can be either 
$\bot$ (undefined) or a natural number indicating the branch chosen in a non-deterministic 
choice (Rule $\rSumCS$).  In each rule, the resulting process
has a new/fresh identifier (see e.g., $j$ in Rule $\rLocalCS$). This new identifier can be obtained, e.g., as the  successor of the  maximal identifier in the previous configuration.  
For the sake of readability,  we write $[i]$ instead of $[i]_\bot$. Moreover, we shall avoid the identifier ``$\ \idxP{i}$'' when it can be inferred from the context. 

\noindent\emph{Stores and Configurations.} The solution for (2) amounts to consider the store, 
in  a configuration,   as a set of (atomic) constraints and not as a constraint. 
Then, the store $\{c_1,\cdots,c_n\}$  represents the constraint $c_1 \sqcup \cdots \sqcup c_n$.

Consider the process  $\tellp{c}$ and let $V \subseteq Vars$.
 The Rule $\rTellCS$ first decomposes the constraint $c$ in  its atoms. For that, assume that the bound variables in $c$ are all distinct and not in $V$ (otherwise, by alpha conversion, we can find $c' \equivC c$ satisfying such condition). We define $atoms(c, V) = \langle bv(c), basic(c)\rangle$
where 

$
basic(c) = \left\{
\begin{array}{l}
c \mbox{ if $c$ is an atom, $\true$, $\false$ \mbox{ or }  $d_{xy}$}\\
basic(c') \mbox{ if } c= \exists x. c' \\
basic(c_1) \cup basic(c_2) \mbox{ if } c= c_1 \sqcup c_2 \\
\end{array}
\right.
$

Observe that in Rule $\rTellCS$, the parameter 
$V$ of the function $atoms$ is the set of free variables occurring in the context, i.e.,  ${\it fvars}$ in Figure \ref{fig:colsem}. This is needed  to perform alpha conversion
of $c$  (which is left implicit in the definition of 
$basic(\cdot)$) to satisfy the above condition on bound names. 

Rule $\rSumCS$ signals the number of the branch $k$ chosen for execution. 
Rule $\rLocalCS$ chooses a fresh variable $x'$, i.e., a variable not in the set of free variables of the configuration (${\it fvars}$).  Hence, we execute the process $P[x'/x]$ and add $x'$ to the set $X$ of local variables. 
Rule $\rCallCS$ is self-explanatory. 

It is worth noticing that we do not consider a rule for structural congruence in the collecting semantics. Such rule, in the system of Figure \ref{fig:sos}, played different roles. Axioms STR2 and  STR3  provide agents with a structure of multiset (commutative and associative). As mentioned above, we consider in the collecting semantics sequences of processes to highlight the process that was reduced in a transition. The sequence $\Gamma$ in Figure \ref{fig:colsem} can be of arbitrary length and then, any of the enabled processes in the sequence can be picked for execution. Axiom STR1 allowed us to perform alpha-conversion on processes. This is needed in  $\rLocal$ to avoid clash of variables. Note that the new Rule $\rLocalCS$ internalizes such procedure by picking a fresh variable $x'$. Finally, Axiom STR4 can be used to simplify $\skipp$ processes that can be introduced, e.g., by a $\rTell$ transition. Observe that the collecting semantics  does not add any  $\skipp$ into the configuration (see Rule $\rTellCS$).

\begin{example}\label{ex:cs}
Consider the following toy example. Let $\mathcal{D}$ contain the process definition 
$A  \defsymboldelta  \tellp{z>x+4} 
$
and
$\mathcal{D}.P$ be a program where 

$
P=\tellp{y<7}  \parallel \whenp{x<0}{A} \parallel \tellp{x=-3}
$. The following is a possible trace generated by the collecting semantics. \\

${\scriptsize \begin{array}{ll}
\qquad \ \  (\emptyset ; \tellp{y<7}\idxP{1},  \whenp{x<0}{A}\idxP{2},\tellp{x=-3}\idxP{3}  ; \true)  \\ \rediIdx{1}    (\emptyset ; \whenp{x<0}{A}\idxP{2}, \tellp{x=-3}\idxP{3} ; y<7) \\
\rediIdx{3}  (\emptyset ;  \whenp{x<0}{A}\idxP{2} ; y<7,x=-3)  
\rediIdxJ{2}{1}  (\emptyset ; {A}\idxP{4} ; y<7,x=-3) \\
\rediIdx{4}  (\emptyset ; \tellp{z>x+4}\idxP{5} ; y<7,x=-3)  
\rediIdx{5}    (\emptyset ; \epsilon ; y<7,x=-3, z>x+4) 
\end{array}
}
$

\end{example}

Now we introduce the notion of observables for the collecting semantics and we show that it coincides with that of Definition \ref{def:obs} for the  operational semantics. 

\begin{definition}[Observables Collecting Semantics]\label{def:obs2}
We write 
$\gamma \rediIdxJ{i_1,...,i_n}{k_1,...,k_n} \gamma'$  whenever
 $
\gamma = (X_0;\Gamma_0;S_0) \rediIdxJ{i_1}{k_1} \cdots 
\rediIdxJ{i_n}{k_n} (X_n;\Gamma_n;S_n) = \gamma'
$. Moreover, if $\exists X_n. \bigsqcup\limits_{d\in S_n}d \entails c$, then we write $\BarbC{\gamma}{c}$. 
If $X_0=S_0=\emptyset$,   we simply write   $\BarbC{\Gamma_0}{c}$.


\end{definition}
\begin{theorem}[Adequacy]
For any process $P$, constraint $c$ and $i \in \mathbb{N}$, $\Barb{P}{c}$ iff $\BarbC{P\idxP{i}}{c}$ 
\end{theorem}
\begin{proof}(\emph{sketch}) ($\Rightarrow$) 
The proof proceeds by induction on the length of the derivation needed to perform the output $c$ in $\Barb{P}{c}$ and using the following results. 

Given a set of variables $V$, a constraint $d$ and a set of constraints $S$, let us use  $\toCollect{d}_V$ to denote (the resulting tuple) $atoms(d ,V)$ and  $\fromCollect{S}_V$ to denote the constraint $\exists V . \bigsqcup\limits_{c_i \in S} c_i$.  If  $\langle Y, S\rangle = \toCollect{d}_{V}$, from the definition of \emph{atoms},  we have 
$d \equivC \fromCollect{S}_Y$.

Let $\Gamma$ (resp. $\Psi$) be a multiset (resp. sequence) of processes.  Let us use 
$\toCollect{\Gamma}$ to denote  any sequence of processes with distinct identifiers built from the processes in   $\Gamma$ and 
$\fromCollect{\Psi}$ to denote  the multiset built from the processes in  $\Psi$. 
Consider now the transition $\gamma = (X ; \Gamma ; d)\redi(X' ; \Gamma' ; d')$. 
Let $\langle Y, S\rangle = \toCollect{d}_{V}$ where  $V= X \cup \fv(\Gamma) \cup \fv(d)$. 
 By choosing the same process reduced in $\gamma$, we can show that 
 there exist $i,k$ s.t. 
  the collecting semantics  mimics the same transition 
as 
$(X \cup Y, \toCollect{\Gamma}, S) \rediIdxJ{i}{k} (X'\cup Y' ; \toCollect{\Gamma''} ; S')$
 where 
$d' \equivP \fromCollect{S'}_{Y'}$ and  $\Gamma'' \equivP \Gamma'$.

The ($\Leftarrow$)  side follows from similar arguments. 
\end{proof}

%

%
%

\subsection{Marking the Store (Step ${\bf S2}$)}\label{sec:step2}
From the final store the user must indicate the symptoms that are 
relevant to the slice that she wants to recompute. For that, she 
must select a set of constraints that considers relevant to 
identify a bug. Normally, these are constraints at the end of a
partial computation, and there are several strategies that one can 
follow to identify them.

Let us suppose that the final configuration in a partial computation is $(X;\Gamma;S)$. 
The symptoms that something is wrong in the program (in the sense 
that the user identifies some unexpected configuration)
may be (and not limited to) the following:
\begin{enumerate}
\item \emph{Causality:} the user
identifies, according to her knowledge, a subset $S' \subseteq S$ that needs to be explained 
(i.e., we need to identify the processes that produced $S'$). 
\item \emph{Variable Dependencies:} The user may identify a set of variables $V\subseteq \fv(S)$ 
whose constraints need to be explored. Then, one would be interested in marking 
the following set of constraints
\[
S_{sliced} = \{ c \in S \mid vars(c) \cap V \neq \emptyset \}
\]

 \item \emph{Unexpected behaviors}: there is a constraint $c$ entailed from the final 
 store that is not expected  from the  intended behavior of the program. 
 Then, one would be interested in marking the following set of constraints:
 \[
S_{sliced} = \bigcup \{S' \subseteq S \mid  \bigsqcup S' \entails c 
\mbox{ and } S' \mbox{ is set minimal} \}
 \]
where ``$S'$ is set minimal'' means that for any $S'' \subset S'$, $S'' \not\entails c$. 
 \item \emph{Inconsistent output}: The final store should be consistent with respect  
 to a given specification (constraint) $c$, i.e., $S$ in conjunction with $c$ must not 
 be inconsistent. In this case, the set of constraints to be marked is:
 \[
 S_{sliced} = \bigcup \{S' \subseteq S \mid \bigsqcup S' \sqcup c \entails \false 
\mbox{ and } S' \mbox{ is set minimal} \}
 \]
where ``$S'$ is set minimal'' means that for any $S'' \subset S'$, 
$S'' \sqcup c \not\entails \false$.

\end{enumerate}

We note that ``set minimality'', in general,  can be expensive to compute. 
However, we believe that  in some practical cases, as shown in the examples in Section \ref{sec:imp},  this is not so heavy. In any case, we can always use  supersets of the minimal ones which are easier to compute but 
less precise for eliminating useless information. 

%
%
%
%
%
%


\subsection{Trace Slice (Step ${\bf S3}$)}
Starting from the set $S_{sliced}$ above  we can define   
a backward slicing step.
We  shall identify, by means of a backward
evaluation, the set of transitions (in the 
original computation) which are necessary for introducing the elements
in $S_{sliced}$. By doing that, we will 
eliminate information  not related to $S_{sliced}$.

\begin{notation}[Sliced Terms]
We shall use the fresh constant symbol $\sliced$ 
to denote an ``irrelevant'' constraint or process. 
Then, for instance, ``$c\sqcup \sliced$'' results from a constraint 
$c\sqcup d$ where $d$ is irrelevant.  Similarly, 
$\whenp{c}{(P \parallel \sliced)} + \sliced$ results from a process of the 
form $\whenp{c}{(P \parallel Q)} + \sum \whenp{c_l}{P_l}$ where $Q$ 
and the summands  in $\sum \whenp{c_l}{P_l}$ are irrelevant. 
We also assume that a sequence $\sliced, \ldots, \sliced$ with any 
number ($\geq 1$) of occurrences of $\sliced$ is 
equivalent to a single occurrence. 

A replacement is either a pair of the shape $[T / i]$  or $[T / c]$. In the first (resp. second) case, 
 the process with identifier $i$  (resp. constraint $c$) is 
replaced with  $T$.  
We shall use $\theta$ to denote a set of replacements
and we  call these sets  as 
``replacing substitutions''.  The  composition of  replacing substitutions $\theta_1$ and $\theta_2$ 
is given by the set union of   $\theta_1$ and 
$\theta_2$,
and is denoted as $\theta_1 \circ \theta_2$. 
If $\Gamma=P_1\idxP{i_i},...,P_n\idxP{i_n}$, for simplicity, we shall write $[\Gamma / j]$ instead of $[P_1 \parallel \cdots \parallel P_n/j]$. Moreover, we shall write, e.g.,  $\whenp{c}{\Gamma}$ instead of $\whenp{c}{(P1 \parallel \cdots \parallel P_n)}$. 

\end{notation}

 Algorithm  \ref{alg:slicer}  computes the slicing. The last configuration in the 
 sliced trace is $(X_n\cap vars(S); \sliced; S)$. 
 This means that we only observe the local variables of interest, 
 i.e., those in $vars(S)$. Moreover, note that  the processes in the 
 last configuration were not executed and then, 
 they are irrelevant (and abstracted  with $\sliced$). 
 Finally, the only relevant constraints are those in $S$.

\begin{algorithm}
\KwIn{- a trace $\gamma_0\rediIdxJ{i_1}{k_1} \cdots
\rediIdxJ{i_{n}}{k_{n}}\gamma_n$ 
where $\gamma_i= (X_i;\Gamma_i; S_i)$

  \qquad\quad\ - a set $S \subseteq S_n$}
    
\KwOut{ a sliced trace $\gamma_0'\redi    \cdots \redi \gamma_n' $}

\Begin{
 {\bf let} $\theta = \emptyset$ {\bf in}
 
 $\gamma_n' \leftarrow ( X_n \cap vars(S) ; \sliced; S)$\;
 \For{l= $n-1$ to 0}{
  $\theta \leftarrow sliceProcess(\gamma_l, \gamma_{l+1}, i_{l+1}, k_{l+1},\theta, S) \circ \theta$\;
  $\gamma_l' \leftarrow (X_l\cap vars(S) ~; ~ \Gamma_l \theta ~;~ S_l \cap S)$
  } 
}
 \caption{Trace Slicer \label{alg:slicer}}
\end{algorithm}

The algorithm backwardly computes the slicing by accumulating 
replacing pairs in $\theta$. The new replacing substitutions are computed by the function 
$sliceProcess$ in 
Algorithm \ref{alg:proc}. 
Suppose that $\gamma \rediIdxJ{i}{k} \psi$. We consider each kind of process. For instance, assume a  $\rTellCS$ transition  $\gamma= (X_\gamma; \Gamma_1, \tellp{c}\idxP{i},\Gamma_2;S_\gamma) \rediIdxJ{i}{} (X_\psi;\Gamma_1,\Gamma_2;S_\psi)=\psi$. 
We note  that $X_\gamma \subseteq X_\psi$ and $S_\gamma \subseteq S_\psi$. 
We replace the constraint $c$ with its sliced version $c'$ computed 
by the function $sliceConstraints$. In that function, we compute the contribution of $\tellp{c}$ to the store, i.e., $S_c = S_\psi \setminus S_\gamma$. Then, any atom $c_a \in S_c$ not in the relevant set of constraints $S$ is 
replaced by $\sliced$. By joining together the resulting atoms, and existentially quantifying the variables in $X_\psi \setminus X_\gamma$ (if any), we obtain the sliced constraint $c'$.  In order to further simplify the trace, if $c'$ is $\sliced$ or $\exists \overline{x}. \sliced$ then we substitute $\tellp{c}$ with $\sliced$ (thus avoiding the ``irrelevant'' process $\tellp{\sliced}$).

 In a non-deterministic choice, all the precluded choices are discarded (`` $+~ \sliced$''). Moreover, if the chosen alternative $Q_k$   does not contribute to the final store (i.e., $\Gamma_Q\theta = \sliced$), then the whole process $\sum\whenp{c_l}{P_l}$ becomes $\sliced$. 
 
Consider the process $\localp{x}{Q}$. Note that $x$ may be replaced  
to avoid a clash of names (see $\rLocal'$).  The (new) created variable must be  
$\{x'\} = X_\psi \setminus X_\gamma$. Then, we check whether $\Gamma_Q[x'/x]$ is 
relevant or not  to return the appropriate replacement. 
 The case of procedure calls can be explained similarly. 

\begin{algorithm}[t]
\SetKwProg{Fn}{Function}{ }{end}
\Fn{sliceProcess($\gamma, \psi, i,k, \theta, S$) }{
  {\bf let} $\gamma=(X_\gamma ; \Gamma,P\idxP{i},\Gamma'; S_\gamma)$ and $\psi=(X_\psi ; \Gamma,\Gamma_Q, \Gamma' ; S_\psi)$ {\bf in}
    \SetKw{KMatch}{match}
  \SetKw{KWith}{with}
  \SetKwBlock{KBMatch}{\KMatch{$P$ \KWith }}{end}

  \KBMatch{
     \uCase{$\tellp{c}$}{
     {\bf let} $c' = sliceConstraints(X_\gamma,X_\psi, S_\gamma, S_\psi, S)$ {\bf in}
     
       \leIf{$c' = \sliced$ or $c' = \exists \overline{x}. \sliced$}{\KwRet{ $[\sliced / i]$}}{\KwRet{ $[\tellp{c'}/i]$}}
       
     } 
     \uCase{$\sum\whenp{c_l}{Q_l}$}{
     \leIf{$ \Gamma_Q \theta=\sliced$}{\KwRet{ $[\sliced / i]$}}{\KwRet{ 
     $[ \whenp{c_k}{(\Gamma_Q\theta)} + \sliced ~/~ i]$}}
     } 
     \uCase{$\localp{x}{Q}$}{
     {\bf let} $\{x'\} = X_\psi \setminus X_\gamma$ {\bf in}
     
       \leIf{$\Gamma_Q[x'/x] \theta=\sliced$}{\KwRet{ $[\sliced / i]$}}{\KwRet{ $[\localp{x'}{\Gamma_Q[x'/x]\theta}/ i]$}}
     } 
     \uCase(){$p(\vy)$}{
     
     	\leIf{$\Gamma_Q \theta=\sliced$}{\KwRet{ $[\sliced / i]$}}{\KwRet{ $\emptyset$}}

     } 
  }

} 

\SetKwProg{Fn}{Function}{ }{end}
\Fn{sliceConstraints($X_\gamma, X_\psi, S_\gamma, S_\psi, S$) }{
{\bf let} $S_c = S_\psi \setminus S_\gamma \mbox{ and } \theta= \emptyset$ {\bf in}

 \lForEach{$c_a \in S_c \setminus S$}{
   $\theta \leftarrow \theta \circ [\sliced / c_a]$
 }
 
 \KwRet{ $\exists_ {X_{\psi} \setminus X_{\gamma}}. \bigsqcup  S_c \theta$}
}

 \caption{Slicing Processes and Constraints \label{alg:proc}}
\end{algorithm}

\begin{example}\label{ex:trace}
Let $a,b,c,d,e$ be constraints without any entailment and consider the  process 
\resizebox{.9\textwidth}{!}{$
R= \whenp{a}{\tellp{c}} \parallel \whenp{c}{(\tellp{d} \parallel  
\tellp{b})} \parallel  \tellp{a} \parallel  \whenp{e}{\skipp}
$
}

In any execution of $R$, the final store is $\{a,b,c,d\}$.  If the user selects 
only $\{d\}$ as slicing criterion, our implementation (see Section \ref{sec:imp}) 
returns the 
following output (omitting the processes' identifiers):
\begin{Verbatim}[fontsize=\scriptsize]
[0; * || ask(c, tell(d) || *) || * || * || * ; *] --> 
[0; * || tell(d) || * || * || * || * ; *] --> 
[0; * || * || * || * || * || * ; d,*] --> 
[0; * || * || * || * || * || * ; d,*] --> stop
\end{Verbatim}

Note that only the relevant part of the process
$\whenp{c}{(\tellp{d} \parallel  \tellp{b})}$ is highlighted as well as the process $\tellp{d}$ 
that introduced $d$ in the final store. 

\end{example}

Also note that the process $P=\whenp{a}{\tellp{c}}$ is not selected in the trace since 
$c$ is not part of the marked store. However, one may be interested in marking this process to discover 
the \emph{causality  relation} 
between $P$ and $Q=\whenp{c}{(\tellp{d} \parallel  \tellp{b})}$. 
Namely,  $P$ adds $c$ to the store, needed in $Q$ to produce $d$.

It turns out that we can easily adapt  Algorithm \ref{alg:proc} to capture such causality relations as follows. 
Assume that $sliceProcess$ returns both, a  replacement $\theta$ and a constraint $c$, i.e., a tuple of 
the shape $\langle \theta, c\rangle$. 
In the case of $\sum\whenp{c_l}{P_l}$, 
if $\Gamma_Q \theta \neq \sliced$,  we return the pair $\langle [ 
\whenp{c_k}{\Gamma_k\theta} + \sliced ~/~ i], 
c_k\rangle$. In all the other cases, 
we return $\langle\theta,  \true\rangle$ where $\theta$ is as in   
Algorithm \ref{alg:proc}. Intuitively, the second component of the tuple represents the guard that 
was entailed in a ``relevant'' application of the rule $\rSumCS$. 
Therefore, in Algorithm \ref{alg:slicer}, besides  accumulating $\theta$, we 
add  the returned guard to  the set of relevant constraints $S$. This is done by replacing the line 5 in 
Algorithm \ref{alg:slicer} with \\

 \ \ \ ${\bf let} \langle \theta', c\rangle  =  
 sliceProcess(\gamma_l, \gamma_{l+1}, i_{l+1}, k_{l+1},\theta, S) \circ \theta\ \ 
 $ {\bf in}

 \ \  \ \ \ \ $\theta  \leftarrow \theta' \circ \theta$

 \ \  \ \ \ \ $S  \leftarrow  S \cup  S_{minimal}(S_{l},c) $
\ \\

where $S_{minimal}(S,c)=\emptyset$ if $c=\true$; otherwise,  $S_{minimal}(S,c) = \bigcup \{S' \subseteq S \mid  \bigsqcup S' \entails c 
\mbox{ and } S' \mbox{ is set minimal} \}$. 
Therefore,  we add to $S$ the minimal set of constraints in $S_k$ that ``explains'' the entailed guard $c$ of  an ask agent.

With this modified version of the algorithm (supporting \emph{causality relations}), the output for the program in Example \ref{ex:trace} is: 
\begin{Verbatim}[fontsize=\scriptsize]
[0 ; ask(a, tell(c)) || ask(c, tell(d) || *) || * || tell(a) || * ; *][3] 
\end{Verbatim}

where the process $\tellp{a}$ is also selected since the execution  of $\whenp{a}{\tellp{c}}$ depends on this process. \\

\noindent{\bf Soundness} 
We conclude here by showing that the slicing procedure 
computes a suitable approximation of the concrete trace. Given two processes $P, P'$, we say that $P'$ approximates $P$, notation 
$P \appr P'$, if there exists a (possibly empty) replacement $\theta$ s.t. $P' = P\theta$ (i.e., $P'$ is as $P$ but replacing some subterms with $\sliced$ ). 
Let $\gamma= (X ; \Gamma  ; S)$ and $\gamma'= (X' ;\Gamma' ; S')$ be two configurations s.t. $|\Gamma| =|\Gamma'|$. We say that $\gamma'$ approximates $\gamma$, notation 
$\gamma \appr \gamma'$, if $X' \subseteq X$, $S' \subseteq S$ and  $P_i \appr P_i'$ for all $i \in  1..|\Gamma|$.  
\begin{theorem}
Let $\gamma_0 \rediIdxJ{i_1}{k_1} \cdots \rediIdxJ{i_n}{k_n}\gamma_n$ be a partial computation and $\gamma'_0 \rediIdxJ{i_1}{k_1} \cdots \rediIdxJ{i_n}{k_n}\gamma'_n$ be the resulting sliced trace according to an arbitrary slicing criterion. Then,   for all $t\in 1..n$, $\gamma_t \appr \gamma_t'$. 
Moreover, let $Q=\sum \whenp{c_k}{P_k}$ and assume that $(X_{t-1}; \Gamma, Q\idxP{i_t}, \Gamma' ; S_{t-1})\rediIdxJ{i_t}{k_t} (X_t; \Gamma, P_{k_t}\idxP{j}, \Gamma'; S_t)$
for some $t \in 1..n$.   If the sliced trace is computed with 
the Algorithm that supports  causality relations,  then 
$\exists X'_{t-1} (\bigsqcup S'_{t-1}) \entails c_{k_t}$.
\end{theorem}

\makeatletter{}
\section{Applications to Timed CCP}\label{sectionapplications}
Reactive systems \cite{BeGo92} are those that react continuously with
their environment at a  rate controlled by the environment. 
For example, a controller or a  signal-processing system, receives a
stimulus (input) from the environment, computes an output  and
then waits for the next interaction with the environment.

Timed \ccp\   (\TCC) \cite{DBLP:journals/jsc/SaraswatJG96,NPV02} 
is an extension of \ccp\ tailoring ideas from  Synchronous Languages
\cite{BeGo92}. More precisely, 
 time in \TCC\ is  conceptually divided into \emph{time
intervals }(or \emph{time-units}). In a particular time
interval, a \ccp\ process $P$  gets an input  $c$
from the environment, it executes with this input as the initial
\emph{store}, and when it reaches
its resting point, it \emph{outputs} the resulting store $d$ to the
environment. The resting point determines also a residual process $Q$
that is then executed in the next time-unit. The resulting store $d$
is not automatically transferred to the next time-unit.
This way, 
outputs of two different time-units are not supposed to be related.

\begin{definition}[Syntax of \TCC\
\cite{DBLP:journals/jsc/SaraswatJG96,NPV02}]\label{tcc:syntax} 
The syntax of \TCC\ is obtained by adding to Definition
\ref{def:syntax-lcc} the processes 
$
\nextp{P} \mid \unlessp{c}{P} \mid \ \bangp{P}
$. 
\end{definition}

The process $\nextp{P}$ delays the execution of $P$ to the next time
interval. 
We shall use  $\nextp\!^n{P}$ to denote $P$ preceded with $n$ copies
of ``$\nextp$'' 
and $\nextp\!^0{P}=P$. 

The \emph{time-out} \( \unlessp{c}{P} \) is also a unit-delay, but
\( P \)  is executed in the next time-unit only if  \( c \) is not
entailed by the final store at the current time interval.

The replication $\bangp{P}$ means  \( P\parallel
\mathbf{next}\,P\parallel \mathbf{next}^{2}P\parallel
\dots\), i.e., unboundedly many copies of \( P \) but one at a time.
We note that in \TCC, recursive calls must be guarded by a {\bf next}
operator to avoid infinite computations during a time-unit. Then,
recursive definitions can be encoded via the  $\bangp{}$ operator
\cite{DBLP:conf/ppdp/NielsenPV02}.

The operational semantics of \TCC\ considers \emph{internal} and
\emph{observable} transitions. The internal transitions 
correspond to the operational steps that take place during a
time-unit. The rules are the same as  in Figure \ref{fig:colsem} plus: \\

\resizebox{.6\textwidth}{!}{
$
\begin{array}{c}
\infer[\rUnless]{(X ; \Gamma, \unlessp{c}{P}\idxP{i},\Gamma' ; S) \rediIdx{i} (X; \Gamma,\Gamma' ; S)}{\bigsqcup S \entails c}
\\\\
\infer[\rBang]{(X ; \Gamma, \bangp{P}, \Gamma'; S) \rediIdx{i} (X ; \Gamma, P\idxP{j} , 
\nextp\bangp{P}\idxP{j'}, \Gamma';S)}{}
\end{array}
$
}
\ \\

where $j$ and $j'$ are fresh identifiers. 
The $\mathbf{unless}$ process is precluded from execution  if its guard can
be entailed from the current store. 
The process $\bangp{P}$ creates a copy of $P$ in the current
time-unit and it is executed in the next time-unit. 
 The seemingly missing rule for the $\mathbf{next}$ operator is 
clarified below. 
 
 The \emph{observable transition} $P \rede{(c,d)} Q $ (``\( P \)
on input \( c \), reduces
in one \emph{time-unit} to \( Q \) and outputs \( d \)'') is  obtained from a finite sequence  of internal reductions: 
\[
\infer[\rObserv]{\Gamma \rede{(c,\exists X. c')}  \localp{X}
F(\Gamma')  }{(\emptyset; \Gamma ; c) \rediIdxJ{i_1,...,i_n}{k_1,...,k_n} (X; \Gamma' ;
c')\not\redi}
\]

The process $F(\Gamma')$ (the continuation of $\Gamma'$) is obtained
as follow:

${
F(R)=\left\{
\begin{array}{ll}
 \skipp & \mbox{ if $R=\skipp$ or $R=\whenp{c}{R'}$ }  \\
 F(R_1) \parallel F(R_2) & \mbox{ if $R = R_1 \parallel R_2$} \\
 Q & \mbox{ if $R=\nextp{Q}$ or $R=\unlessp{c}{Q}$ } 
\end{array}
\right.
}
$

The function $F(R)$ (the future of $R$) 
returns the processes that must be executed in the next time-unit. More precisely, 
it unfolds \emph{next} and
$\emph{unless}$ expressions. Notice that an \emph{ask} process
reduces to $\skipp$ if its  guard was not entailed by the final
store. Notice also  that $F$ is not defined for $\tellp{c}$,
$\bangp{Q}$, $\localp{x}{P}$ or $p(\vx)$ processes since all of them
give rise to an internal transition. 
Hence these processes can only appear in the continuation if they 
occur within a $\mathbf{next}$ or $\mathbf{unless}$ expression.

\subsection{A trace Slicer for \TCC} \label{sec:imp}
From the execution point of view, only the observable transition is
relevant since it describes
the input-output behavior of processes. 
However, when a \TCC\ program is debugged, we have to consider also
the internal transitions. This makes the task of debugging even
harder when compared to \ccp. 

We implemented in Maude (\url{http://maude.cs.illinois.edu}) a
prototypical version of a slicer for \TCC\ (and then for \ccp) that
can be found at \url{http://subsell.logic.at/slicer/}. 

The slicing technique for the internal transition is based on the
Algorithm \ref{alg:slicer}
by  adding the following cases to Algorithm \ref{alg:proc}:
 
 \begin{algorithm}[H]
 \lCase{$\unlessp{c}{Q}$}{
     {\KwRet{ $[\sliced/i]$}}  
     }
 \uCase{$\bangp{Q}$}{
     \leIf{$\Gamma_Q \theta=\sliced\ $}{\KwRet{ $[\sliced / i]$}}{\KwRet{
$[\bangp{(Q\theta)}/i]$}}}
 \end{algorithm} 
 \ \\
 Note that if an {\bf unless} process evolves during a time-unit, then
it is irrelevant. In the case of $\bangp{P}$, 
we note that $\Gamma_Q = Q\idxP{j}, \nextp\bangp{Q}\idxP{j'}$. 
We check whether $P$ is
relevant in the current time-unit ($Q$) or in the following one
($\nextp\bangp Q$). If this is not the case, then $\bangp{Q}$
is irrelevant. 

Recall that {\bf next} processes do not exhibit any
transition during a time-unit and then, we do not consider this case
in the  extended version of Algorithm \ref{alg:proc}. 

 For the observable transition we proceed as follows. Consider a
trace of $n$ observable steps $\gamma_0 \rede{} \cdots \rede{}
\gamma_n$ and a set $S_{slice}$ of relevant constraints to be
observed in the last configuration $\gamma_n$. Let $\theta_n$   be
the replacement computed during the slicing process of the
(internal) trace generated from $\gamma_n$.   We propagate the
 replacements in $\theta_n$  to the configuration $\gamma_{n-1}$ as
follows:
\begin{enumerate}
 \item In $\gamma_{n-1}$ we set $S_{sliced} = \emptyset$. Note that
the unique store of interest for the user is the one in $\gamma_n$.
Recall also that the final store in \TCC\ is not transferred to the
next time-unit. Then, only the processes (and not the constraints) in
$\gamma_{n-1}$  are responsible for the final store in $\gamma_n$. 
 \item 
 Let $\psi$ be the last internal configuration in $\gamma_{n-1}$,
i.e., $\gamma_{n-1} \rediIdxJ{i_1,...,i_m}{k_1,...,k_m} \psi \not\redi$ and $\gamma_n = F(\psi)$. 
 We propagate the replacements in $\theta_n$ to 
 $\psi$ before running the slicer on the trace  starting from
$\gamma_{n-1}$. For that, we compute a replacement $\theta'$ that
must be applied to $\psi$ as follows:
 \begin{itemize}
  \item If there is a process $R = \nextp{P}\idxP{i}$ in $\psi$, then
$\theta'$ includes the replacement 
$[\nextp{(\Gamma_P\theta_n)}/i]$. For instance, if 
  $R = \nextp{(\tellp{c}\parallel \tellp{d})}$  and $\tellp{c}$ was
irrelevant in $\gamma_n$, the resulting process in $\psi$ is 
 $\nextp{( \sliced \parallel \tellp{d})}$ . The case
for  $\unlessp{c}{P}$ is similar. 
 \item If there is a process $R =\sum_l\whenp{c_l}{P_l} \idxP{i}$ in $\psi$ (which is
irrelevant since it was not executed), we  add to $\theta'$ the
 replacement $[\sliced / i]$. 
 \end{itemize}
 \item 
 Starting from $\psi \theta$, we compute the slicing on
$\gamma_{n-1}$ (Algorithm \ref{alg:slicer}).
 
\item  This procedure continues until the first configuration  $\gamma_0$ is reached. 
\end{enumerate}


\begin{example}
Consider the following process definitions:
\[
\begin{array}{lll}
System \defsymbol  Beat2 \parallel Beat4 \qquad\qquad
Beat2 \defsymbol \tellp{b2} \parallel \nextp^2 \ Beat2
\\
Beat4 \defsymbol \tellp{b4} \parallel \nextp^4 \ Beat4
\end{array}
\]
This is a simple model of a multimedia system that, every 2 (resp. 4) time-units, produces the constraint $b2$ (resp. $b4$). Then, 
every 4 time-units, the system produces
both $b2$ and $b4$. 
  If we compute  5 time-units and choose $S_{slice}=\{b4\}$ we obtain (omitting the process identifiers):
\begin{Verbatim}[fontsize=\scriptsize]
{1 / 5 > [System ; *] -->  [Beat4 ; *] --> [next^4(Beat4) ; *]} ==> 
{2 / 5 > [next^3(Beat4) ; *]} ==> 
{3 / 5 > [next^2(Beat4) ; *] } ==> 
{4 / 5 > [next(Beat4) ; *]} ==> 
{5 / 5 > [Beat4 ; *] --> [tell(b4) || * ; *] --> [* ; b4]}
\end{Verbatim}
Note that all the executions of $Beat2$ in time-units 1, 3 and 5 are
hidden  since they do not contribute to the observed output $b4$.
More interestingly, the execution of $\tellp{b4}$ in time-unit 1, as
well as the recursive call of $Beat4$ ($\nextp^4 \ Beat4$) in
time-unit 5, are also hidden. 

Now assume that we compute an even number of time-units. Then, no
constraint is produced in that time-unit and the whole execution of
$System$ is hidden:
\begin{Verbatim}[fontsize=\scriptsize]
{1/4 > [* ; *]} ==> {2/4 > [* ; *]} ==> 
{3/4 > [* ; *]} ==> {4/4 > [* ; *]}
\end{Verbatim}

As a more compelling example, consider the following process
definitions:

\resizebox{.95\textwidth}{!}{
$
\begin{array}{lcl}
Beat   \defsymbol   \prod\limits_{i\in I_1}\nextp^i{\tellp{\beatc} }
& \qquad&
Start  \defsymbol   \sum\limits_{i\in I_2} \nextp^i(\tellp{\startc})
\\
Check   \defsymbol  \bangp
\whenp{\startc}{\nextp^{12}(\tellp{\stopc}}) 
& \qquad&
System   \defsymbol  Beat \parallel Start \parallel Check
\end{array}
$}
\ \\

where $I_1 = \{0,3,5,7,9,11,14,16,18,20,22\}$, $I_2=\{0,3,5,7,9,11\}$ and $\Pi_i$ stands for parallel composition. 
This process represents a rhythmic pattern where groups of ``$2$''-unit elements separate groups of ``$3$''-unit elements, e.g., $
3\ \underbrace{\ 2\ 2\ 2\ 2\ }\ 3\ \underbrace{\ 2\ 2\ 2\ 2 \ 2}$. Such pattern appears in repertoires of Central African Republic music \cite{chemillier-book} and were programmed in \TCC\ in \cite{cp-music}. 

This pattern can be  represented in a circle with $24$ divisions,
where ``$2$'' and ``$3$''-unit elements are placed. The ``$3$''-unit intervals are displayed in red in Figure
\ref{fig:circ}. The important property is \emph{asymmetry}: 
if one attempts to break the circle into two parts, it is not
possible to have two equal parts. To be more precise, the $\startc$
and $\stopc$ constraints divide the circle in two halves (see process Start) and
it is always the case that the constraint  $\beatc$ does not coincide
in a time-unit with the constraint $\stopc$. For instance, in Figure \ref{fig:circles} (a) (resp. (b)), the circle is divided in time-units 1 --start-- to 13 --stop-- (resp. 4 --start-- to 16 --stop--).  The signal $\beatc$ does not coincide with a $\stopc$: in Figure  \ref{fig:circles} (a) (resp. (b)), the $\beatc$ is added in time-unit 12 (resp. 15). 

If we generate one of the possible traces and perform the slicing processes  for the time-unit 13 with
$S_{sliced}=\{\beatc, \stopc\}$, we only observe  as relevant process
$Check$ (since no $\beatc$ is produced in that time-unit) :
\begin{Verbatim}[fontsize=\scriptsize]
{1 / 13 > [System ; *] --> [Check ; *] --> [! ask(start, next^12(tell(stop)) ; *] 
          --> [ask(start, next^12(tell(stop)) ; *] --> [next^12(tell(stop) ; *]} ==> 
.... ==> ...
{11 / 13 > [next(next(tell(stop))) ; *]} ==> 
{12 / 13 > [next(tell(stop)) || * ; *]} ==> 
{13 / 13 > [tell(stop) ; *] --> [* ; stop][0]}
\end{Verbatim}
More interestingly, assume that we wrongly write a process $Check$ that
is not ``well synchronized'' with the process $Beat$. For instance,
let  $I_2'=\{2\}$. In this case,  the $\startc$ signal does not coincide with a
$\beatc$.  Then, in time-unit 15, we (wrongly) observe both $\beatc$ and
$\stopc$ (i.e., asymmetry  is broken!). The trace of that program (that can be found in tool's web
page) is quite long and difficult to understand. On the contrary, the
sliced one is rather simple:
\begin{Verbatim}[fontsize=\scriptsize]
{1 / 15 > [System ; *] --> [Beat || Check ; *] --> 
          [next^14(tell(beat) || next(! ask(start, next^12(tell(stop)); *]} ==> 
{2 / 15 > [next^13(tell(beat))|| ! ask(start, next^12(tell(stop))) ; *]} ==> 
{3 / 15 > [next^12(tell(beat)))|| ! ask(start, next^12(tell(stop)) ; *]} ==> 
{4 / 15 > [next^11(tell(beat))|| next^11(tell(stop)|| * ; *] --> stop} ==> 
...
{14 / 15 > [next(tell(beat)) || next(tell(stop)) || * ; *] --> stop} ==> 
{15 / 15 > [tell(beat) || tell(stop) || * ; *] --> [tell(stop) || * ; beat] --> 
           [* ; beat,stop]}
\end{Verbatim}
Something interesting in this trace is that the {\bf ask} in the $Check$ process  is hidden from the time-unit 4 on (since it is not ``needed'' any more). Moreover, the only $\tellp{\beatc}$ process (from $Beat$ definition) displayed is the one that is executed in time-unit 15 (i.e., the one resulting from $\nextp^{14}{\tellp{\beatc}}$). From this trace, it is not difficult to note that the $Start$ process starts on time-unit 3 (the process $\nextp^{11}{\tellp{\stopc}}$ first appears on time-unit $4$). This can  tell the user that the process $Start$ begins its execution in a wrong time-unit. In order to confirm this hypothesis, the user may compute the sliced trace up to time-unit 3 with  
$S_{sliced}=\{\beatc, \startc\}$ and notice that, in that time-unit, $\startc$ is produced but $\beatc$ is not part of the store. 
\end{example}

The reader may find in the web page of the tool a further example related to biochemical systems. 
We modeled in \TCC\ the P53/Mdm2 DNA-damage Repair Mechanism \cite{DBLP:conf/fmmb/MariaDF14}. The slicer allowed us to detect two bugs in the written code. We invite the reader to check in this example the length (and complexity)  of the buggy trace and the resulting sliced trace. 

\begin{figure}[t]
\centering
\begin{tabular}{cc}
\includegraphics[width=5.5cm]{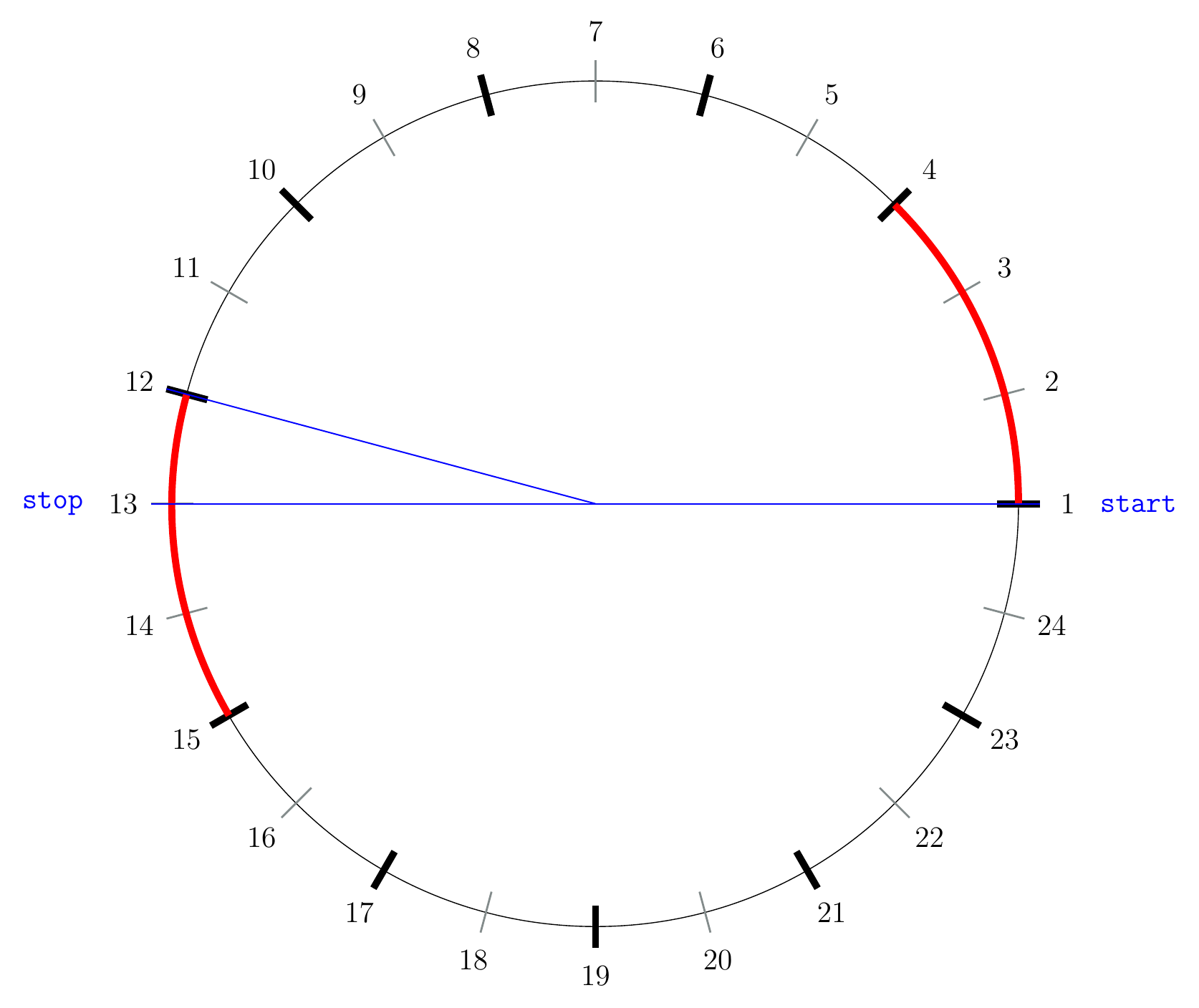} &
\includegraphics[width=4.6cm]{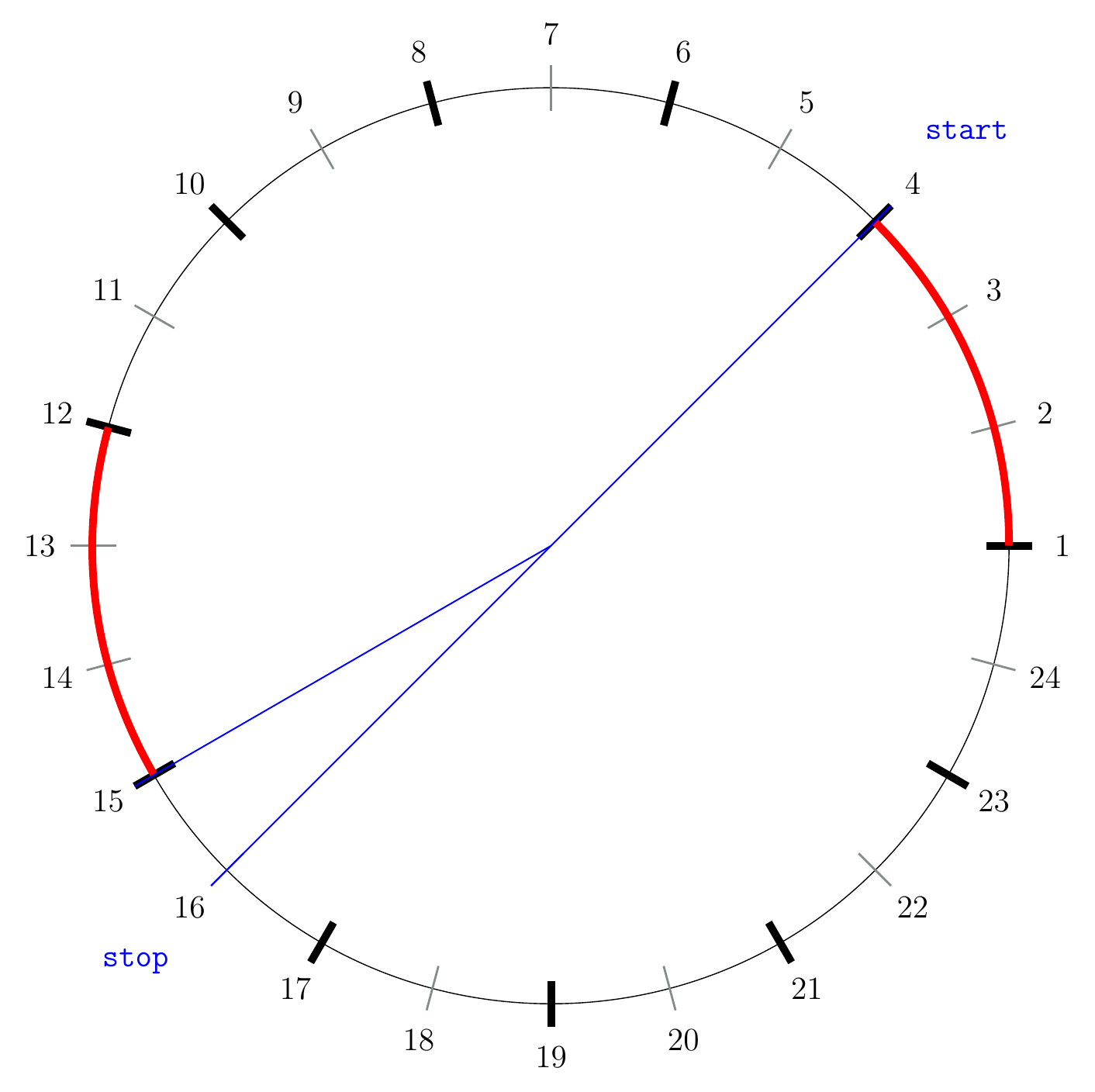} \\
(a) & (b)
\end{tabular}
\caption{\label{fig:circ}Pattern of ``$2$'' and ``$3$''-unit elements
(taken from \cite{chemillier-book}).  \label{fig:circles}}
\end{figure}

\makeatletter{}
\section{Conclusions and future work}\label{sectionconclusions}

In this paper we introduced the first framework for slicing
 concurrent constraint  based programs, and showed its applicability for CCP and timed CCP.  We 
  implemented a  prototype of the slicer in Maude
and showed its use in debugging a program specifying  a biochemical system and a multimedia interacting system.

 Our framework is a good basis for dealing with other variants of \ccp\
such as  linear \ccp\ \cite{fages01ic},  spatial and epistemic \ccp\
\cite{DBLP:conf/concur/KnightPPV12} as well as with other temporal
extensions of it \cite{DBLP:journals/iandc/BoerGM00}. We are 
currently working on extending our tool to cope with these languages. 
We  also plan to incorporate into our framework an assertion language based on a suitable fragment of temporal logic. Such assertions will specify invariants the program must satisfy  during its execution. 
If the assertion is not satisfied in a given state, then the execution is interrupted and a concrete trace is generated to be later  sliced. 
For instance,  in the multimedia system,  the user may specify the invariant $\stopc \to  (\neg \beatc)$ (if $\stopc$ is entailed then $\beatc$ cannot be part of the store) or $\stopc \to \past \beatc$ (a $\stopc$ state must be preceded by a $\beatc$ state).


\noindent {\em\bf Acknowledgments}. We thank the anonymous reviewers for 
their detailed comments and suggestions which helped us to improve our paper. 
The work of Olarte was funded by CNPq (Brazil). 


\bibliographystyle{plain}

\begin{thebibliography}{10}

\bibitem{ABER11}
M.~Alpuente, D.~Ballis, J.~Espert, and D.~Romero.
\newblock Backward trace slicing for rewriting logic theories.
\newblock In {\em Proc. of CADE'11}, pages 34--48, Berlin, Heidelberg, 2011.
  Springer-Verlag.

\bibitem{ABFR14}
M.~Alpuente, D.~Ballis, F.~Frechina, and D.~Romero.
\newblock Using conditional trace slicing for improving maude programs.
\newblock {\em Sci. Comput. Program.}, 80:385--415, 2014.

\bibitem{BeGo92}
G.~Berry and G.~Gonthier.
\newblock The {{\sc Esterel}} synchronous programming language: Design,
  semantics, implementation.
\newblock {\em Science of Computer Programming}, 19(2):87--152, 1992.

\bibitem{bortolussi}
L.~Bortolussi and A.~Policriti.
\newblock Modeling biological systems in stochastic concurrent constraint
  programming.
\newblock {\em Constraints}, 13(1-2):66--90, 2008.

\bibitem{chemillier-book}
M.~Chemillier.
\newblock {\em Les Math\'ematiques Naturelles}.
\newblock Odile Jacob, 2007.

\bibitem{CFM94}
M.~Codish, M.~Falaschi, and K.~Marriott.
\newblock Suspension {A}nalyses for {C}oncurrent {L}ogic {P}rograms.
\newblock {\em ACM Transactions on Programming Languages and Systems},
  16(3):649--686, 1994.

\bibitem{CominiTV11absdiag}
M.~Comini, L.~Titolo, and A.~Villanueva.
\newblock {Abstract Diagnosis for Timed Concurrent Constraint programs}.
\newblock {\em Theory and Practice of Logic Programming}, 11(4-5):487--502,
  2011.

\bibitem{DBLP:journals/iandc/BoerGM00}
F.~S. de~Boer, M.~Gabbrielli, and M.~C. Meo.
\newblock A timed concurrent constraint language.
\newblock {\em Inf. Comput.}, 161(1):45--83, 2000.

\bibitem{BoerPP95}
F.~S. de~Boer, A.~Di Pierro, and C.~Palamidessi.
\newblock Nondeterminism and infinite computations in constraint programming.
\newblock {\em Theoretical Computer Science}, 151(1):37--78, 1995.

\bibitem{fages01ic}
F.~Fages, P.~Ruet, and S.~Soliman.
\newblock Linear concurrent constraint programming: Operational and phase
  semantics.
\newblock {\em Inf. Comput.}, 165(1):14--41, 2001.

\bibitem{DBLP:journals/tplp/FalaschiOP15}
M.~Falaschi, C.~Olarte, and C.~Palamidessi.
\newblock Abstract interpretation of temporal concurrent constraint programs.
\newblock {\em {TPLP}}, 15(3):312--357, 2015.

\bibitem{HentenryckSD98}
P.~Van Hentenryck, V.~A. Saraswat, and Y.~Deville.
\newblock Design, implementation, and evaluation of the constraint language
  cc(fd).
\newblock {\em Journal of Logic Programming}, 37(1-3):139--164, 1998.

\bibitem{Silva2012}
S.~Josep.
\newblock A vocabulary of program slicing-based techniques.
\newblock {\em ACM Comput. Surv.}, 44(3):12:1--12:41, June 2012.

\bibitem{DBLP:conf/concur/KnightPPV12}
S.~Knight, C.~Palamidessi, P.~Panangaden, and F.~D. Valencia.
\newblock Spatial and epistemic modalities in constraint-based process calculi.
\newblock In M.~Koutny and I.~Ulidowski, editors, {\em CONCUR}, volume 7454 of
  {\em LNCS}, pages 317--332. Springer, 2012.

\bibitem{KL88}
B.~Korel and J.~Laski.
\newblock Dynamic program slicing.
\newblock {\em Inf. Process. Lett.}, 29(3):155--163, 1988.

\bibitem{DBLP:conf/fmmb/MariaDF14}
E.~De Maria, J.~Despeyroux, and A.~P. Felty.
\newblock A logical framework for systems biology.
\newblock In F.~Fages and C.~Piazza, editors, {\em {FMMB}}, volume 8738 of {\em
  LNCS}, pages 136--155. Springer, 2014.

\bibitem{DBLP:conf/ppdp/NielsenPV02}
M.~Nielsen, C.~Palamidessi, and F.~D. Valencia.
\newblock On the expressive power of temporal concurrent constraint program.
  languages.
\newblock In {\em Proc. of PPDP'02}, pages 156--167. {ACM}, 2002.

\bibitem{NPV02}
M.~Nielsen, C.~Palamidessi, and F.~D. Valencia.
\newblock Temporal concurrent constraint programming: Denotation, logic and
  applications.
\newblock {\em Nord. J. Comput.}, 9(1):145--188, 2002.

\bibitem{OSV08}
C.~Ochoa, J.~Silva, and G.~Vidal.
\newblock Dynamic slicing of lazy functional programs based on redex trails.
\newblock {\em Higher Order Symbol. Comput.}, 21(1-2):147--192, June 2008.

\bibitem{DBLP:journals/tcs/OlartePN15}
C.~Olarte, E.~Pimentel, and V.~Nigam.
\newblock Subexponential concurrent constraint programming.
\newblock {\em Theor. Comput. Sci.}, 606:98--120, 2015.

\bibitem{cp-music}
C.~Olarte, C.~Rueda, G.~Sarria, M.~Toro, and F.~D. Valencia.
\newblock Concurrent constraints models of music interaction.
\newblock In G.~Assayag and C.~Truchet, editors, {\em Constraint Programming in
  Music}, pages 133--153. Wiley, 2011.

\bibitem{DBLP:journals/constraints/OlarteRV13}
C.~Olarte, C.~Rueda, and F.~D. Valencia.
\newblock Models and emerging trends of concurrent constraint programming.
\newblock {\em Constraints}, 18(4):535--578, 2013.

\bibitem{Olarte:08:SAC}
C.~Olarte and F.~D. Valencia.
\newblock Universal concurrent constraint programing: symbolic semantics and
  applications to security.
\newblock In R.~L. Wainwright and H.~Haddad, editors, {\em SAC}, pages
  145--150. ACM, 2008.

\bibitem{cp-book}
V.~A. Saraswat.
\newblock {\em Concurrent Constraint Programming}.
\newblock MIT Press, 1993.

\bibitem{DBLP:journals/jsc/SaraswatJG96}
V.~A. Saraswat, R.~Jagadeesan, and V.~Gupta.
\newblock Timed default concurrent constraint programming.
\newblock {\em J. Symb. Comput.}, 22(5/6):475--520, 1996.

\bibitem{saraswat91popl}
V.~A. Saraswat, M.~C. Rinard, and P.~Panangaden.
\newblock Semantic foundations of concurrent constraint programming.
\newblock In D.~S. Wise, editor, {\em POPL}, pages 333--352. ACM Press, 1991.

\bibitem{Shapiro83}
E.~Y. Shapiro.
\newblock {\em Algorithmic Program DeBugging}.
\newblock MIT Press, 1983.

\bibitem{MW84}
M.~Weiser.
\newblock Program slicing.
\newblock {\em IEEE Trans. on Software Engineering}, 10(4):352--357, 1984.

\end{thebibliography}


\end{document}